\documentclass[aps,pra,reprint,superscriptaddress]{revtex4-2}
\usepackage{amsmath,amssymb,amsthm,mathtools,braket,bm,color,hyperref,graphicx}
\hypersetup{colorlinks=true,linkcolor=blue,citecolor=blue,urlcolor=blue}

\newtheorem{theorem}{Theorem}
\newtheorem{proposition}{Proposition}
\newtheorem{lemma}{Lemma}
\newtheorem{corollary}{Corollary}
\newtheorem{remark}{Remark}

\begin{document}

\title{Complete product-state contact set and optimality of the canonical three-qubit Shifts witness}

\author{Bing-Dong Wan}
\affiliation{School of Physics and Electronic Technology, Liaoning Normal University, Dalian 116029, China}
\affiliation{Center for Theoretical and Experimental High Energy Physics, Liaoning Normal University, Dalian 116029, China}
\author{Yu-Chen Guo}
\email{ycguo@lnnu.edu.cn}
\affiliation{School of Physics and Electronic Technology, Liaoning Normal University, Dalian 116029, China}
\affiliation{Center for Theoretical and Experimental High Energy Physics, Liaoning Normal University, Dalian 116029, China}

\begin{abstract}
For the three-qubit Shifts unextendible product basis (UPB), the minimum product-state expectation value $\lambda_{\mathrm{Shifts}}=1-3\sqrt6/8$ of the UPB projector is known, and four attaining product states were exhibited previously. We determine the complete equality set over all complex product states and prove that it consists of exactly eight product rays: two rays invariant under cyclic permutation of the three qubits and two cyclic orbits of size three. The proof combines exact phase analysis, a Gr\"obner-basis classification of all 28 finite real stationary points, and a complete treatment of the projective boundaries. Crucially, the eight contact rays span the full three-qubit Hilbert space, whereas the four previously known rays do not. The complete equality classification therefore establishes that the associated projector witness satisfies the spanning criterion and is optimal in the sense that no nonzero positive semidefinite operator can be subtracted while preserving block positivity. As a corollary, for the white-noise family $\rho(p)=(1-p)\rho_{\mathrm{Shifts}}+p\mathbb I/8$, where $\rho_{\mathrm{Shifts}}$ is the complementary-projector Shifts state, the same witness detects entanglement for $p<2-3\sqrt6/4\approx16.29\%$.
\end{abstract}

\keywords{bound entanglement, unextendible product basis, optimal entanglement witness, spanning property, product-state contact set}

\maketitle

\section{Introduction}
\label{sec:intro}

Bound entanglement---entanglement that cannot be distilled into pure maximally entangled states by local operations and classical communication (LOCC)---exhibits the distinction between entanglement and distillability~\cite{Horodecki1998,Peres1996,Horodecki1996,HorodeckiRMP2009,Hiesmayr2025review}. Bound-entangled states can nevertheless contain distillable secure key~\cite{HorodeckiKey2005}. Their preparation and certification have been demonstrated experimentally~\cite{Amselem2009,Lavoie2010}, while robust verification remains an active problem~\cite{Sentis2018}.

A central detection tool is the entanglement witness---a Hermitian operator $W$ whose expectation value is non-negative on all separable states but negative on at least one entangled state~\cite{Terhal2000,Lewenstein2000,GuhneToth2009,ChruscinskiSarbicki2014,Bruss2002}. Throughout this work, ``separable'' means fully separable, and block positivity means nonnegativity on all pure fully product vectors. The witness considered below therefore detects failure of full separability, rather than genuine multipartite entanglement. Optimality---the property that no nonzero positive semidefinite operator can be subtracted while retaining block positivity---is controlled by the zero-expectation product states, the \emph{contact set} of $W$. If these states span the full Hilbert space, the witness is optimal by the spanning criterion of Lewenstein \emph{et al.}~\cite{Lewenstein2000}, a criterion further analyzed for three-qubit witnesses in Ref.~\cite{Kye2015}. Testing it, however, requires the \emph{complete} set of minimizing product states, which is known explicitly only in special cases; when the underlying product-state optimization cannot be solved analytically, one must fall back on numerical methods~\cite{Wiesniak2020}.

Unextendible product bases (UPBs) offer a constructive route to bound-entangled states~\cite{Bennett1999upb,DiVincenzo2003,Bravyi2004,Johnston2014}. A UPB is an incomplete set of mutually orthogonal product states spanning a subspace whose orthogonal complement contains no product state. For a UPB $\mathcal S=\{|\psi_i\rangle\}_{i=1}^n$, the normalized projector onto that complement has positive partial transpose (PPT) across every bipartition and is entangled. The associated projector witness, introduced in the early UPB literature and further developed in subsequent constructions~\cite{Bennett1999upb,Terhal2000,BGR2005}, is $W=P_{\mathcal S}-\lambda\mathbb I$, where $P_{\mathcal S}=\sum_i|\psi_i\rangle\langle\psi_i|$ projects onto the UPB subspace and
\begin{equation}
\lambda = \min_{\substack{|\phi\rangle\ \mathrm{product}\\ \langle\phi|\phi\rangle=1}}
\langle\phi|P_{\mathcal S}|\phi\rangle
=\min_{\substack{|\phi\rangle\ \mathrm{product}\\ \langle\phi|\phi\rangle=1}}
\sum_i |\langle\phi|\psi_i\rangle|^2 > 0
\label{eq:lambda-def}
\end{equation}
is the minimum expectation value of the UPB projector over normalized product states $|\phi\rangle=|\phi_A\rangle|\phi_B\rangle\cdots$. For the three-qubit Shifts UPB, Branciard \emph{et al.} proved analytically that
\begin{equation}
\lambda_{\mathrm{Shifts}}=1-\frac{3\sqrt6}{8}
\label{eq:known-lambda}
\end{equation}
and exhibited four product states attaining equality~\cite{Branciard2010}, in the course of evaluating the geometric entanglement~\cite{WeiGoldbart2003} of the complementary-projector state. The associated Shifts state is separable with respect to every fixed bipartition but is not fully separable~\cite{Bennett1999upb,Branciard2010,Acin2001,SeevinckUffink2008}, and its four-state basis has minimum cardinality among three-qubit orthogonal UPBs~\cite{Bravyi2004,Johnston2014}.

Yet beyond the value of $\lambda_{\mathrm{Shifts}}$ and these four exemplary minimizers, the equality structure of the optimization has remained unexplored: the complete equality set was not classified, and---crucially for applications---the four known rays span only a four-dimensional subspace, too small to test the spanning criterion. The optimality of the canonical Shifts projector witness has therefore remained an open question. More broadly, complete contact-set classifications are rare for UPB witnesses, and optimality has typically been presumed on numerical grounds rather than proved exactly.

In this work we close this gap. We determine the complete equality set over all complex product states and prove that it consists of exactly eight product rays: two rays invariant under cyclic permutation of the three qubits and two cyclic orbits of size three. We prove completeness through exact phase-equality analysis, classification of all 28 finite real stationary points, and complete treatment of the projective boundaries. The eight contact rays---unlike the four previously known ones---span the full Hilbert space, establishing witness optimality by the spanning criterion~\cite{Lewenstein2000}, and we obtain the white-noise detection threshold as a direct corollary of the previously known coefficient.

The remainder of this paper is organized as follows. Section~\ref{sec:setup} introduces the Shifts UPB and its projector witness and reduces the complex optimization to real amplitudes. Section~\ref{sec:result} reviews the known coefficient and derives the symmetric equality rays. Section~\ref{sec:global} establishes the complete equality classification and its white-noise corollary, while Sec.~\ref{sec:spanning} proves the spanning property and witness optimality. Section~\ref{sec:discuss} discusses the implications and outlook. Appendices~\ref{app:phase}--\ref{app:reproduce} provide the phase-equality analysis, the Gr\"obner-basis and projective-boundary details, the spanning determinant, and details of the symbolic verification.

\section{Shifts UPB and projector witness}
\label{sec:setup}

Let $\mathcal{H} = (\mathbb{C}^2)^{\otimes 3}$ with $D = 8$. In the convention used here, the Shifts UPB~\cite{Bennett1999upb,Bravyi2004} consists of $n=4$ orthogonal product states
\begin{equation}
\mathcal{S}_{\mathrm{Shifts}} = \big\{\,|0,1,+\rangle,\; |1,+,0\rangle,\; |+,0,1\rangle,\; |-,-,-\rangle\,\big\},
\label{eq:shifts-upb}
\end{equation}
where $|\pm\rangle = (|0\rangle \pm |1\rangle)/\sqrt{2}$. We denote the four states in Eq.~\eqref{eq:shifts-upb}, in the displayed order, by $\{|\psi_i\rangle\}_{i=1}^4$. This representation is locally unitarily equivalent to the Shifts convention used by Branciard \emph{et al.}~\cite{Branciard2010}. The complementary subspace has dimension $D-n = 4$ and contains no product state; hence
\begin{equation}
\rho_{\mathrm{Shifts}} = \tfrac{1}{4}\bigl(\mathbb{I} - \textstyle\sum_{i=1}^{4}|\psi_i\rangle\langle\psi_i|\bigr)
\label{eq:rho-shifts}
\end{equation}
is separable across every fixed bipartition but not fully separable, and is therefore bound entangled~\cite{Bennett1999upb,Branciard2010}.

The canonical projector witness~\cite{Bennett1999upb,Terhal2000,BGR2005} is $W = P_{\mathcal{S}} - \lambda\,\mathbb{I}$, where $P_{\mathcal{S}} = \sum_{i=1}^{4}|\psi_i\rangle\langle\psi_i|$ and $\lambda$ is defined in Eq.~\eqref{eq:lambda-def}. Here and throughout, the optimization is over normalized fully product states $|\phi\rangle = |a\rangle|b\rangle|c\rangle$ of the three qubits (labeled $A,B,C$).

\subsection{Complex phase elimination}
\label{sec:complex}

\begin{lemma}[Phase-vertex reduction]
\label{lem:real}
For the Shifts UPB, the minimum in Eq.~\eqref{eq:lambda-def} is attained at real amplitudes: it suffices to consider $|a\rangle,|b\rangle,|c\rangle \in \mathbb{R}^2$.
\end{lemma}

\begin{proof}
Parameterize each qubit as $|q_j\rangle = (|0\rangle + z_j|1\rangle)/\sqrt{1+|z_j|^2}$ with $z_j = r_j e^{i\phi_j}$ ($r_j \ge 0$, $\phi_j \in [0,2\pi)$). The omitted ray $|1\rangle$ is recovered by continuity as $r_j\to\infty$, or equivalently in the complementary projective chart. For $F(\mathbf r,\boldsymbol{\phi})=\langle\phi|P_{\mathcal S}|\phi\rangle$, a direct computation gives
\begin{equation}
F(\mathbf{r},\boldsymbol{\phi})=
\frac{\mathcal{F}(\mathbf{r},\boldsymbol{\phi})}
{8\prod_{j=1}^3(1+r_j^2)},
\label{eq:F-complex}
\end{equation}
where
\begin{equation}
\begin{aligned}
\mathcal{F}={}&4r_2^2(1+r_3^2+2r_3\cos\phi_3)\\
&+4r_1^2(1+r_2^2+2r_2\cos\phi_2)\\
&+4r_3^2(1+r_1^2+2r_1\cos\phi_1)\\
&+\prod_{j=1}^3(1+r_j^2-2r_j\cos\phi_j).
\end{aligned}
\end{equation}
For fixed radii $r_1,r_2,r_3$, let $N\equiv\mathcal F$. It is a multilinear function of $x_j = \cos\phi_j$:
\[
N = A + \sum_j B_j x_j + \sum_{i<j} C_{ij} x_i x_j + D_{123} x_1 x_2 x_3,
\]
where $A$, $B_j$, $C_{ij}$, and $D_{123}$ are real coefficients depending only on the radii. For fixed $x_2,x_3$, this expression is affine in $x_1$ and therefore has a minimum at $x_1=\pm1$. Repeating the argument for $x_2$ and $x_3$ proves that a phase minimum is attained at a vertex of $[-1,1]^3$. Equivalently,
\begin{equation}
N(\mathbf{x})=\sum_{\boldsymbol{\sigma}\in\{\pm1\}^3}
N(\boldsymbol{\sigma})\prod_{j=1}^3\frac{1+\sigma_jx_j}{2},
\label{eq:multilinear-interpolation}
\end{equation}
so every interior value is a convex combination of vertex values. Hence $\min_{\boldsymbol{\phi}}F=\min_{x_j\in\{\pm1\}}F$, corresponding to real amplitudes.
\end{proof}

By Lemma~\ref{lem:real}, we may restrict to real product states. Write
\[
\begin{aligned}
|a\rangle&=\cos a\,|0\rangle+\sin a\,|1\rangle,\\
|b\rangle&=\cos b\,|0\rangle+\sin b\,|1\rangle,\\
|c\rangle&=\cos c\,|0\rangle+\sin c\,|1\rangle.
\end{aligned}
\]
and introduce the tangent variables $u=\tan a$, $v=\tan b$, and $w=\tan c$. The objective then takes the compact rational form
\begin{equation}
\begin{split}
f(u,v,w)={}&\Bigl[4v^2(1+w)^2+4u^2(1+v)^2
+4w^2(1+u)^2\\
&+(1-u)^2(1-v)^2(1-w)^2\Bigr]\\
&\big/\Bigl[8(1+u^2)(1+v^2)(1+w^2)\Bigr],
\end{split}
\label{eq:f-of-t-general}
\end{equation}
which is manifestly invariant under the cyclic permutation $(u,v,w)\mapsto(v,w,u)$. This permutation generates the cyclic group $C_3=\langle(ABC)\rangle$; below, ``$C_3$-symmetric'' or ``$C_3$-fixed'' refers to invariance under this action.

\section{Known coefficient and symmetric equality rays}
\label{sec:result}

\begin{proposition}[Symmetric equality rays]
\label{prop:symmetric}
The restriction of Eq.~\eqref{eq:lambda-def} to the $C_3$-symmetric product states has two minima, $|\phi^*\rangle = |a^*\rangle^{\otimes 3}$ and $|\phi^{*\prime}\rangle = |a^{*\prime}\rangle^{\otimes 3}$, where $|a^*\rangle = \cos a^*\,|0\rangle + \sin a^*\,|1\rangle$ and $|a^{*\prime}\rangle = \sin a^*\,|0\rangle + \cos a^*\,|1\rangle$ with
\begin{equation}
\tan a^* = 2 + \sqrt{6} + \sqrt{9 + 4\sqrt{6}}\,.
\label{eq:a-star}
\end{equation}
Both rays attain the previously known minimum~\cite{Branciard2010}:
\begin{equation}
\boxed{\;\lambda_{\mathrm{Shifts}} = 1 - \frac{3\sqrt{6}}{8} \;\approx\; 0.08144134645630821\;}
\label{eq:lambda-closed}
\end{equation}
(the smaller root of $32\lambda^2 - 64\lambda + 5 = 0$). They form the two $C_3$-fixed members of the complete equality set derived below.
\end{proposition}

\begin{proof}[Derivation from the $C_3$-symmetric submanifold]
The $C_3$-symmetric submanifold $u=v=w=t$ reduces the objective to
\begin{equation}
f_{\mathrm{sym}}(t) = \frac{12t^2(t+1)^2 + (t-1)^6}{8(t^2+1)^3}.
\label{eq:f-of-t}
\end{equation}
Differentiating and factoring the numerator of $f_{\mathrm{sym}}'(t)$:
\begin{equation}
f_{\mathrm{sym}}'(t) \propto (t-1)(t+1)(t^4 - 8t^3 - 6t^2 - 8t + 1).
\label{eq:critical}
\end{equation}
The quartic $t^4 - 8t^3 - 6t^2 - 8t + 1 = 0$ has two real roots, $t_+ = 2+\sqrt{6}+\sqrt{9+4\sqrt{6}}$ and $t_- = 2+\sqrt{6}-\sqrt{9+4\sqrt{6}}$, satisfying $t_+ t_- = 1$. The remaining real critical points $t=\pm 1$ yield $f_{\mathrm{sym}} = 3/4$ and $f_{\mathrm{sym}} = 1$, both larger than $\lambda_{\mathrm{Shifts}}$, while the projective endpoint satisfies $\lim_{t\to\infty}f_{\mathrm{sym}}(t)=1/8>\lambda_{\mathrm{Shifts}}$. Substituting $t_\pm$ gives $f_{\mathrm{sym}}(t_\pm) = 1 - 3\sqrt{6}/8$.

\end{proof}

\begin{remark}
The reciprocal roots $t_+t_-=1$ define two distinct contact rays of the fixed projector $P_{\mathcal S}$. Their equal overlap is an algebraic degeneracy of the objective and does not follow from a local symmetry that preserves $P_{\mathcal S}$.
\end{remark}

\section{Complete equality classification}
\label{sec:global}

Proposition~\ref{prop:symmetric} recovers two symmetric rays attaining the known minimum. To determine the complete equality set, we classify the nonsymmetric stationary points and the projective boundaries.

\subsection{Gr\"obner basis elimination}

The objective $f(u,v,w)$ of Eq.~\eqref{eq:f-of-t-general} is a rational function, smooth on $\mathbb{R}^3$ and bounded ($f\in[0,1]$). We compute the Gr\"obner basis~\cite{Cox2015} of the stationary system $\{\partial_u f, \partial_v f, \partial_w f\}=0$ in lexicographic order $w \succ v \succ u$, which eliminates $v$ and $w$ and yields a single elimination polynomial $P(u)=0$ that any $u$-component of a stationary point must satisfy. This polynomial factors completely over $\mathbb{Q}$:
\begin{widetext}
\begin{align}
P(u) =\;& u\,(u-1)(u+1)\,(u^2-3)\,(3u^2-1)\,(u^2-4u+1)\,(u^2+4u+1)\nonumber\\
&\times\,(3u^4+6u^2-5)\,(5u^4-6u^2-3)\nonumber\\
&\times\,(u^4-8u^3-6u^2-8u+1)\,(u^4+8u^3-6u^2+8u+1).
\label{eq:elim-poly}
\end{align}
\end{widetext}
Its 19 real roots, $u \in \{0,\, \pm 1,\, \pm\sqrt{3},\, \pm 1/\sqrt{3},\, 2\pm\sqrt{3},\, -(2\pm\sqrt{3}),\, \pm\tfrac{\alpha}{3}\ (\text{with }\alpha=\sqrt{6\sqrt{6}-9}),\, \pm\sqrt{\tfrac{3+2\sqrt{6}}{5}},\, \pm t_+,\, \pm t_-\}$, enumerate all possible $u$-coordinates of stationary points.

For each real root $u^*$, the remaining Gr\"obner basis elements $G_1(v,u)=0$ and $G_2(w,v,u)=0$ determine the companion variables $(v,w)$. Solving these equations yields the following complete classification of stationary points in the main projective chart $|0\rangle + u|1\rangle$ for each qubit:

\begin{enumerate}
\item \textbf{$C_3$-symmetric stationary points (4 points).} $u=v=w \in \{1, -1, t_+, t_-\}$. The function values are $f(1,1,1)=3/4$, $f(-1,-1,-1)=1$, and $f(t_\pm,t_\pm,t_\pm)=\lambda_{\mathrm{Shifts}}$.

\item \textbf{Minimum orbit (6 points, $f = \lambda_{\mathrm{Shifts}}$).} Two cyclic orbits of size 3:
\begin{align}
(u,v,w) &= \bigl(\tfrac{\alpha}{3},\; -t_+,\; -\tfrac{3}{\alpha}\bigr), \label{eq:orbit1}\\
(u,v,w) &= \bigl(-\tfrac{\alpha}{3},\; -t_-,\; \tfrac{3}{\alpha}\bigr), \label{eq:orbit2}
\end{align}
and their cyclic permutations. Exact substitution into the three gradient numerators gives $\nabla f=\mathbf{0}$ and $f=\lambda_{\mathrm{Shifts}}$.

\item \textbf{Lower nonminimum level (9 points).} Three cyclic orbits, generated by
\begin{align}
&\bigl(-\sqrt3,-1/\sqrt3,1\bigr),
\quad\bigl(-(2-\sqrt3),0,1/\sqrt3\bigr),\nonumber\\
&\bigl(-1,2-\sqrt3,2+\sqrt3\bigr),
\label{eq:nonmin-low}
\end{align}
all have
\begin{equation}
f_-=\frac{7-3\sqrt3}{16}\approx0.112740>\lambda_{\mathrm{Shifts}}.
\end{equation}

\item \textbf{Upper nonminimum level (9 points).} Three cyclic orbits, generated by
\begin{align}
&\bigl(\sqrt3,1/\sqrt3,1\bigr),
\quad\bigl(0,-1/\sqrt3,-(2+\sqrt3)\bigr),\nonumber\\
&\bigl(-1,2+\sqrt3,2-\sqrt3\bigr),
\label{eq:nonmin-high}
\end{align}
all have
\begin{equation}
f_+=\frac{7+3\sqrt3}{16}\approx0.762260>\lambda_{\mathrm{Shifts}}.
\end{equation}
\end{enumerate}

Thus the finite real stationary set contains exactly $4+6+9+9=28$ points. The remaining triangular Gr\"obner-basis equations determine the companion coordinates for every real root of $P(u)$ and produce precisely the cyclic orbits listed above. Appendix~\ref{app:groebner} gives the zero-dimensionality, quotient dimension, real-root count, and the completeness argument.

\subsection{Boundary analysis}

The real projective space $\mathbb{RP}^1$ for each qubit is covered by two charts: $|q\rangle \propto |0\rangle + u|1\rangle$ ($u\in\mathbb{R}$) and $|q\rangle \propto s|0\rangle + |1\rangle$ ($s\in\mathbb{R}$, related by $u=1/s$). The above analysis covers the chart where all three qubits use the first parameterization. Stationary points in charts where one or more qubits use the second parameterization are obtained by the substitution $u \to 1/u$ (and similarly for $v,w$). The key additional candidates are the boundaries where $u\to\infty$ (or $v\to\infty$, $w\to\infty$), which correspond to $s=0$ in the alternative chart. We analyze three cases:

\begin{enumerate}
\item \textbf{One variable $\to\infty$ (e.g., $u\to\infty$, $v,w$ finite).} The limit is
\begin{equation}
f_\infty^{(1)}(v,w) = \frac{4(1+v)^2 + 4w^2 + (1-v)^2(1-w)^2}{8(1+v^2)(1+w^2)}.
\label{eq:boundary-one-function}
\end{equation}
The stationary equations on this projective face give the values
\begin{equation}
\left\{\frac{5\pm\sqrt{13}}{16},\frac{7\pm3\sqrt3}{16},\frac34,1\right\}.
\end{equation}
After also checking the boundaries of this face, its global minimum is
\begin{equation}
\min_{v,w}f_\infty^{(1)}(v,w)=\frac{5-\sqrt{13}}{16}
\approx0.087153>\lambda_{\mathrm{Shifts}}.
\label{eq:boundary-one}
\end{equation}
It is attained at the two pairs $(v,w)$ satisfying $w=v+2$ with
\begin{equation}
v=\frac{1+\sqrt{13}\pm\sqrt{30+6\sqrt{13}}}{2}.
\end{equation}

\item \textbf{Two variables $\to\infty$ (e.g., $u,v\to\infty$, $w$ finite).} The limit is
\begin{equation}
f_\infty^{(2)}(w) = \frac{4 + (1-w)^2}{8(1+w^2)}.
\end{equation}
Differentiation gives stationary points $w=2\pm\sqrt5$, and hence
\begin{equation}
\min_w f_\infty^{(2)}(w)=\frac{3-\sqrt5}{8}
\approx0.095492>\lambda_{\mathrm{Shifts}},
\label{eq:boundary-two}
\end{equation}
attained at $w=2+\sqrt5$.

\item \textbf{All three variables $\to\infty$.} $f \to 1/8>\lambda$.
\end{enumerate}

The cyclic symmetry covers the other choices of one or two infinite coordinates. Combining the stationary-point classification with the boundary analysis, we obtain:

\begin{theorem}[Complete complex equality set]
\label{thm:main}
For normalized product states $|\phi\rangle$, the equality condition
\begin{equation}
\langle\phi|P_{\mathcal S}|\phi\rangle
=\lambda_{\mathrm{Shifts}}
\end{equation}
holds for exactly eight complex product rays: the two $C_3$-fixed rays in Proposition~\ref{prop:symmetric} and the two cyclic orbits of size three in Eqs.~\eqref{eq:orbit1}--\eqref{eq:orbit2}.
\end{theorem}

\begin{proof}
The lower bound $\lambda_{\mathrm{Shifts}}$ is known from Ref.~\cite{Branciard2010}. Lemma~\ref{lem:real} reduces the search for equality to real phase vertices. The exact Gr\"obner-basis analysis gives 28 finite real stationary points: 2 symmetric equality rays, 2 other symmetric points, 6 nonsymmetric equality rays, and 18 nonsymmetric nonminimizers. Appendix~\ref{app:groebner} gives the zero-dimensional completeness argument. Equations~\eqref{eq:boundary-one}--\eqref{eq:boundary-two}, together with Appendix~\ref{app:boundary}, exclude equality on every projective boundary. If a complex product state attains equality, Eq.~\eqref{eq:multilinear-interpolation} expresses its value as a convex combination equal to the global minimum; consequently, every sign vertex with nonzero interpolation weight must itself be a real equality point. The real classification therefore restricts its radii to the four representatives in Appendix~\ref{app:phase}, where the minimizing sign vertex is unique. Hence no additional complex product ray attains equality.
\end{proof}

\begin{corollary}[Witness-detection threshold]
For $W=P_{\mathcal S}-\lambda_{\mathrm{Shifts}}\mathbb I$ and
$\rho(p)=(1-p)\rho_{\mathrm{Shifts}}+p\mathbb I/8$, one has
\begin{equation}
\operatorname{Tr}[W\rho(p)]=-\lambda_{\mathrm{Shifts}}+\frac{n}{D}p.
\end{equation}
Hence this witness detects $\rho(p)$ for $p<p_W$, where
\begin{equation}
\boxed{p_W=\frac{D}{n}\lambda_{\mathrm{Shifts}}
=2-\frac{3\sqrt6}{4}\approx0.162883.}
\label{eq:pW-closed}
\end{equation}
At $p=p_W$ the witness expectation vanishes.
\end{corollary}

\section{Spanning property and optimality of the witness}
\label{sec:spanning}

Define the product-state contact set
\begin{equation}
\mathcal Z(W)=\bigl\{|\phi\rangle\text{ product}:\langle\phi|W|\phi\rangle=0\bigr\}.
\end{equation}
Theorem~\ref{thm:main} identifies $\mathcal Z(W)$ completely, allowing the spanning criterion to be tested exactly.

\begin{proposition}[Full-span contact set]
\label{prop:spanning}
The eight minimizing product rays span the full three-qubit Hilbert space:
\[
\operatorname{span}\{|\Phi_k\rangle\}_{k=1}^8 = (\mathbb{C}^2)^{\otimes 3},
\]
where $\{|\Phi_1\rangle,\ldots,|\Phi_8\rangle\}$ are the unnormalized product vectors corresponding to the eight minimizers.
\end{proposition}

\begin{proof}
Construct the $8\times8$ matrix $M$ whose columns are the unnormalized vectors $|\Phi_j\rangle=(1,u_j)^T\otimes(1,v_j)^T\otimes(1,w_j)^T$, expressed in the ordered computational basis $\{|000\rangle,|001\rangle,\ldots,|111\rangle\}$. Order the columns as the two symmetric minimizers, followed by the three cyclic images of Eq.~\eqref{eq:orbit1} and then those of Eq.~\eqref{eq:orbit2}. Exact symbolic reduction gives
\begin{equation}
\det M=-\frac{196608}{5}\bigl(39+16\sqrt6\bigr)\neq0.
\label{eq:determinant}
\end{equation}
Therefore $\operatorname{rank}(M)=8$, and the complete contact set contains a basis of the full space.
\end{proof}

\begin{corollary}[Spanning-property optimality]
\label{cor:optimal}
The canonical projector witness $W = P_{\mathcal{S}} - \lambda_{\mathrm{Shifts}}\,\mathbb{I}$ satisfies the spanning criterion~\cite{Lewenstein2000}: the eight minimizing product states lie on the zero-expectation surface $\langle\Phi_k|W|\Phi_k\rangle = 0$ and span the full Hilbert space. Consequently, $W$ is an \emph{optimal entanglement witness}: no nonzero positive semidefinite operator can be subtracted, with a positive coefficient, while preserving block positivity.
\end{corollary}

\begin{proof}
Because the vectors are unnormalized, the minimizing condition reads
\begin{equation}
\langle\Phi_k|P_{\mathcal S}|\Phi_k\rangle
=\lambda_{\mathrm{Shifts}}\langle\Phi_k|\Phi_k\rangle,
\end{equation}
and hence $\langle\Phi_k|W|\Phi_k\rangle=0$. Suppose that $W-\epsilon Q$ remained block positive for some $Q\ge0$ and $\epsilon>0$. Evaluation on each contact vector gives
\begin{equation}
0\le\langle\Phi_k|(W-\epsilon Q)|\Phi_k\rangle
=-\epsilon\langle\Phi_k|Q|\Phi_k\rangle,
\end{equation}
so $\langle\Phi_k|Q|\Phi_k\rangle=0$. Positivity of $Q$ then implies $Q|\Phi_k\rangle=0$ for every $k$. Since Proposition~\ref{prop:spanning} shows that these vectors span the full space, $Q=0$. Thus no nonzero positive operator can be subtracted, which proves optimality directly and recovers the spanning criterion~\cite{Lewenstein2000} in the present multipartite fully separable setting.
\end{proof}

\section{Discussion and outlook}
\label{sec:discuss}

Branciard \emph{et al.} established the value of $\lambda_{\mathrm{Shifts}}$ and exhibited the four equality states mapped to Eq.~\eqref{eq:branciard-four}~\cite{Branciard2010}. Those four vectors span only a four-dimensional subspace. The present result adds the four rays in Eq.~\eqref{eq:additional-four} and proves that the resulting eight-ray set is complete. The exact determinant in Eq.~\eqref{eq:determinant} then shows that the complete contact set contains a basis of the three-qubit Hilbert space, converting the equality classification into an optimality proof for the canonical projector witness.

Two clarifications are in order. First, the white-noise value in Eq.~\eqref{eq:pW-closed} is the detection threshold of this particular optimal witness, not a claim about the exact separability threshold of the noisy state. Second, optimality in the Lewenstein sense used here---no nonzero positive semidefinite operator can be subtracted while preserving block positivity---does not by itself imply maximal white-noise tolerance among all witnesses detecting the same state. For many bipartite tile-UPB states, Wie\'sniak \emph{et al.}~\cite{Wiesniak2020} found that state-tailored witnesses obtained from approximations to the closest separable state provide stronger detection than the corresponding UPB projector witnesses. This state-specific geometric optimization is distinct from the spanning-property optimality established here. The present construction therefore provides an analytically exact witness that is optimal in the Lewenstein spanning sense, together with its closed-form white-noise threshold; whether a different witness yields a larger white-noise tolerance remains an open numerical question.

A natural extension is the GenShifts family~\cite{DiVincenzo2003} on $2k-1$ qubits, for which the analogous elimination problem involves $2k-1$ variables and a Gr\"obner-basis computation whose complexity grows rapidly; the $k=3$ (five-qubit) case is the next case beyond the present $k=2$ one where the exact strategy may still be feasible, while a general solution will likely require additional symmetry reduction. 
\begin{acknowledgments}
This work was supported by the National Natural Science Foundation of China (NSFC) Grants No.~12575106, No.~12675136, and No.~12147214, and by the Liaoning Province Fundamental Research Fund No.~LJ212410165019. The authors used GLM-5.2, Kimi K3, and Aether (v0.73) for limited assistance with manuscript revision, code generation and debugging, and scientific consistency checks. The authors specified the physical assumptions and exact validation criteria, independently checked all algebraic outputs, and take full responsibility for the work.
\end{acknowledgments}

\section*{Data Availability Statement}
No experimental or numerical data sets were generated or analyzed in this study. The symbolic calculations reported here follow directly from the equations in the manuscript; the auxiliary scripts used for verification are available from the corresponding authors upon reasonable request.

\appendix

\section{Complete phase-equality conditions}
\label{app:phase}

For comparison with Ref.~\cite{Branciard2010}, let $P_{\rm B}$ denote the Shifts projector in the convention used there. The single-qubit unitary $U=HX$ acts as
\begin{equation}
U|0\rangle=|-\rangle,\quad U|1\rangle=|+\rangle,
\quad U|+\rangle=|0\rangle,
\quad U|-\rangle=-|1\rangle,
\end{equation}
and maps the two conventions according to
\begin{equation}
P_{\mathcal S}=U^{\otimes3}P_{\rm B}U^{\dagger\otimes3},
\end{equation}
up to the ordering and irrelevant phases of the UPB vectors. For a real local state parameterized there by the Bloch angle $\theta$, the transformed tangent coordinate is $u=\tan(\theta/2-\pi/4)$. For the four equality states exhibited in Ref.~\cite{Branciard2010}, $\cos\theta_0=(2-\sqrt6)/2$. Here $\operatorname{cyc}(u,v,w)=\{(u,v,w),(v,w,u),(w,u,v)\}$ denotes the three-element orbit generated by cyclic permutation of the coordinates. Applying the transformation gives
\begin{equation}
\bigl\{(t_-,t_-,t_-)\bigr\}
\cup\operatorname{cyc}\!\left(\frac{\alpha}{3},-t_+,-\frac{3}{\alpha}\right).
\label{eq:branciard-four}
\end{equation}
The four additional contact rays found here are
\begin{equation}
\bigl\{(t_+,t_+,t_+)\bigr\}
\cup\operatorname{cyc}\!\left(-\frac{\alpha}{3},-t_-,\frac{3}{\alpha}\right).
\label{eq:additional-four}
\end{equation}
Thus the new equality information is the second four-ray subset and the proof that the union is complete.

Equation~\eqref{eq:multilinear-interpolation} not only proves that a real phase vertex attains the minimum, but also characterizes equality. Modulo cyclic permutations, the eight real minimizing rays have four radius representatives. Exact evaluation of their eight sign vertices gives
\begin{center}
\begin{tabular}{c|c}
\hline\hline
radius representative & unique minimizing sign vertex \\
\hline
$(t_+,t_+,t_+)$ & $(+,+,+)$ \\
$(t_-,t_-,t_-)$ & $(+,+,+)$ \\
$(\alpha/3,t_+,3/\alpha)$ & $(+,-,-)$ \\
$(\alpha/3,t_-,3/\alpha)$ & $(-,-,+)$ \\
\hline\hline
\end{tabular}
\end{center}
At the displayed vertex the value is $\lambda_{\mathrm{Shifts}}$; every other sign vertex has a strictly positive algebraic difference from $\lambda_{\mathrm{Shifts}}$. Since the coefficients in Eq.~\eqref{eq:multilinear-interpolation} are nonnegative and sum to one, equality at an interior phase point would require support only on minimizing vertices. The minimizing vertex is unique in every row, so $\cos\phi_j=\pm1$ for all three parties. Cyclic invariance gives the same conclusion for the remaining four rays. Hence the eight real rays listed in the main text are also the complete set of complex minimizing rays.

\section{Completeness of the stationary-point classification}
\label{app:groebner}

Let $g_u,g_v,g_w\in\mathbb Q[u,v,w]$ be the numerators of the three derivatives of Eq.~\eqref{eq:f-of-t-general}; their denominators are strictly positive on $\mathbb R^3$. For the ideal
\begin{equation}
I=\langle g_u,g_v,g_w\rangle
\end{equation}
and lexicographic order $w\succ v\succ u$, exact computation gives the leading-monomial ideal
\begin{equation}
\operatorname{LM}(I)=
\langle w,v^3,u^2v^2,u^7v,u^{27}\rangle.
\label{eq:appendix-leading}
\end{equation}
The standard monomials are
\begin{equation}
\{u^i:0\le i\le26\}\cup
\{u^iv:0\le i\le6\}\cup
\{u^iv^2:0\le i\le1\},
\end{equation}
and therefore
\begin{equation}
\dim_{\mathbb Q}\bigl(\mathbb Q[u,v,w]/I\bigr)=27+7+2=36.
\label{eq:appendix-length}
\end{equation}
Thus $I$ is zero dimensional and has algebraic length 36 over $\mathbb C$, counted with multiplicity.

The last polynomial in the lexicographic basis is the eliminant $P(u)$ in Eq.~\eqref{eq:elim-poly}. It has degree 27, is square free, and an exact Sturm count gives 19 real and eight nonreal roots. The 28 triples listed in Sec.~\ref{sec:global} are distinct exact real zeros of $I$. Because a finite projection of a zero-dimensional variety is closed, the elimination theorem implies that each of the eight nonreal eliminant roots lifts to at least one nonreal zero of $I$. Hence $I$ has at least $28+8=36$ distinct complex zeros. Equation~\eqref{eq:appendix-length} bounds the total algebraic length by 36, so all these zeros are simple and no additional zero of $I$ exists. In particular, there is no additional real stationary point. This proves the completeness of the 28-point real classification without relying on numerical root searches.

\section{Exact projective-boundary analysis}
\label{app:boundary}

For the one-infinity face, let $F_1(v,w)$ denote the function in Eq.~\eqref{eq:boundary-one-function}. The numerators of its two derivatives are
\begin{align}
h_v={}&v^2w^2-2v^2w-3v^2-4vw^2-w^2+2w+3,\nonumber\\
h_w={}&v^2w^2-4v^2w-v^2-2vw^2-8vw+2v+w^2-1.
\end{align}
Their exact resultants factor as
\begin{align}
\operatorname{Res}_w(h_v,h_w)
={}&-16v(v-1)(v^2+4v+1)\nonumber\\
&\times(v^4-2v^3-20v^2-18v+3),\label{eq:appendix-resultant-v}\\
\operatorname{Res}_v(h_v,h_w)
={}&16w(w+1)(w^2-3)\nonumber\\
&\times(w^4-10w^3+16w^2+6w-9).
\label{eq:appendix-resultant-w}
\end{align}
After clearing the positive denominator, the numerator of the following stationary-value polynomial lies in the ideal $\langle h_v,h_w\rangle$:
\begin{equation}
\begin{split}
0={}&(F_1-1)(4F_1-3)(64F_1^2-40F_1+3)\\
&\times(128F_1^2-112F_1+11).
\end{split}
\end{equation}
Therefore every finite stationary value belongs to
\begin{equation}
\left\{1,\frac34,\frac{5\pm\sqrt{13}}{16},
\frac{7\pm3\sqrt3}{16}\right\}.
\end{equation}
The projective boundary of this face is the two-infinity problem. Its derivative numerator is $w^2-4w-1$, giving the minimum $(3-\sqrt5)/8$ at $w=2+\sqrt5$. The remaining corner has value $1/8$. Exact comparison with the finite minimum gives
\begin{align}
\frac{5-\sqrt{13}}{16}-\lambda_{\mathrm{Shifts}}
&=\frac{6\sqrt6-11-\sqrt{13}}{16}>0,\\
\frac{3-\sqrt5}{8}-\lambda_{\mathrm{Shifts}}
&=\frac{3\sqrt6-5-\sqrt5}{8}>0,\\
\frac18-\lambda_{\mathrm{Shifts}}
&=\frac{3\sqrt6-7}{8}>0.
\end{align}
This exhausts all projective boundary strata by cyclic symmetry.

\section{Spanning determinant and symbolic verification}
\label{app:reproduce}

The determinant in Eq.~\eqref{eq:determinant} is obtained by exact reduction in the number field generated by $\sqrt6$, $\sqrt{9+4\sqrt6}$, and $\alpha$; in particular, no floating-point arithmetic enters the rank proof.

The exact calculations were performed with Python 3.14.4 and SymPy 1.14.0~\cite{Meurer2017}. Using exact symbolic arithmetic, we constructed the lexicographic Gr\"obner basis in the order $w\succ v\succ u$, verified Eq.~\eqref{eq:appendix-leading} and the quotient dimension, factored and Sturm-counted the eliminant, checked all exact stationary representatives and phase-vertex inequalities, reproduced the boundary resultants and stationary-value polynomial, and verified the exact spanning determinant. All calculations were specified and independently verified by the authors.

\setlength{\bibsep}{2pt}
\bibliographystyle{apsrev4-2}

\end{document}